\documentclass[letterpaper,11pt]{article}

\usepackage{verbatim}
\usepackage{latexsym}
\usepackage{amsmath,amssymb,amsfonts}
\usepackage{fullpage}
\usepackage{epsfig,ifpdf,graphics}

\newenvironment{description*}%
  {\vspace{-1ex}\begin{description}%
    \setlength{\itemsep}{-0.5ex}%
    \setlength{\parsep}{0pt}}%
  {\end{description}}
\newenvironment{itemize*}%
  {\vspace{-1ex}\begin{itemize}%
    \setlength{\itemsep}{-0.5ex}%
    \setlength{\parsep}{0pt}}%
  {\end{itemize}}
\newenvironment{enumerate*}%
  {\vspace{-1ex}\begin{enumerate}%
    \setlength{\itemsep}{-0.5ex}%
    \setlength{\parsep}{0pt}}%
  {\end{enumerate}}

\newcommand\symmdiff{\triangle}

\newcommand\drop[1]{}
\newcommand{\Pq}{P^{\,q}}
\newcommand{\Pql}{P^{\,q}_{\ell}}

\newcommand{\Plc}{P_{\ell,c}}

\newcommand\choicesl{{\textnormal{choices}_{\ell}}}
\newcommand\choiceslc{{\textnormal{choices}_{\ell,c}}}
\newcommand\choicesql{{\textnormal{choices}^{q}_{\ell}}}

\newcommand\mq{{\ceil{q\ell}}}

\newcommand\req[1]{(\ref{#1})}
\newcommand\tref[1]{~\ref{#1}}

{\makeatletter
 \gdef\xxxmark{%
   \expandafter\ifx\csname @mpargs\endcsname\relax 
     \expandafter\ifx\csname @captype\endcsname\relax 
       \marginpar{xxx}
     \else
       xxx 
     \fi
   \else
     xxx 
   \fi}
 \gdef\xxx{\@ifnextchar[\xxx@lab\xxx@nolab}
 \long\gdef\xxx@lab[#1]#2{{\bf [\xxxmark #2 ---{\sc #1}]}}
 \long\gdef\xxx@nolab#1{{\bf [\xxxmark #1]}}
}

\newtheorem{theorem}{Theorem}
\newtheorem{lemma}[theorem]{Lemma}

\newtheorem{corollary}[theorem]{Corollary}

\newtheorem{definition}[theorem]{Definition}

\newcommand{\qed}{\hbox{\rule{6pt}{6pt}}}
\newenvironment{proof}[1][]{\paragraph{Proof{#1}}}{\hfill\qed\medskip\\}

\newcommand{\tO}{\widetilde{O}}

\newcommand{\eps}{\varepsilon}

\newcommand{\fct}{\rightarrow}

\newcommand{\A}{\Phi}
\newcommand{\Ain}{\Phi}
\newcommand{\Aout}{\Psi}
\newcommand{\hi}{h}
\newcommand\cconf{\phi}
\newcommand\Dconf{\gamma}

\newcommand{\ceil}[1]{\lceil {#1}\rceil}

\let\phi=\varphi

\newcommand{\calH}{\mathcal{H}}

\title{Simple Tabulation, Fast Expanders, Double Tabulation, and High Independence\footnote{This paper was published in the Proceedings of the 54nd IEEE Symposium on Foundations of Computer Science (FOCS'13), pages 90--99, 2013 \cite{Tho13:simple-simple}. \textcopyright IEEE.}}

\author{Mikkel Thorup\footnote{Research 
supported in part by an Advanced Grant from the Danish Council for Independent Research under the Sapere Aude research carrier programme.
Part of this research was done while the author was at AT\&T Labs--Research.}\\
University of Copenhagen\\
\texttt{mikkel2thorup@gmail.com}}

\begin{document}

\maketitle

\begin{abstract} Simple tabulation dates back to Zobrist in 1970 who used it for game playing programs.  Keys are viewed as consisting of $c$ characters from some alphabet $\Ain$. We
initialize $c$ tables $h_0, \dots, h_{c-1}$ mapping characters to random
hash values. A key $x=(x_0, \dots, x_{c-1})$ is hashed to $h_0[x_0] \oplus
\cdots \oplus h_{c-1}[x_{c-1}]$, where $\oplus$ denotes bit-wise exclusive-or.
The scheme is extremely fast when the character hash tables $h_i$ are
in cache.
Simple tabulation hashing is not even 4-independent, but we show here 
that if we apply it twice, then we do get high independence. 
First we hash to some intermediate keys that are 6 times longer than
the original keys, and then we hash the intermediate keys to the final
hash values.

The intermediate keys have $d=6c$ characters from $\Ain$. We can then
view the hash function as a highly unbalanced 
bipartite graph with keys on
one side, each with edges to $d$ output characters on the other
side. We show that this graph has nice expansion properties, and from
that it follows that if we perform another level of simple tabulation
on the intermediate keys, then the composition is a highly independent
hash function. More precisely, the independence we get is
$|\Ain|^{\Omega(1/c)}$. In our $O$-notation, we view both $|\Ain|$ and $c$ is going to
infinity, but with $c$ much smaller than $|\Ain|$.

Our space is $O(c|\Ain|)$ and the hash function is evaluated in
$O(c)$ time.  Siegel [FOCS'89, SICOMP'04] has proved that with this
space, if the hash function is evaluated in $o(c)$ time, then the
independence can only be $o(c)$, so our evaluation time is best
possible for $\Omega(c)$ independence---our independence is much
higher if $c=|\Ain|^{o(1/c)}$.  

Siegel used $O(c)^c$ evaluation time to get the same independence with similar
space. Siegel's main focus was $c=O(1)$,  but
we are exponentially faster when $c=\omega(1)$. 

Applying  our scheme recursively, we can increase our independence to
$|\Ain|^{\Omega(1)}$ with $o(c^{\log c})$ evaluation time.
Compared with Siegel's scheme this is both faster and higher independence.

Siegel states about his scheme that it is 
``far too slow for any practical application''. Our scheme 
is trivial to implement, and it does provide realistic implementations of 100-independent hashing for, say,
32-bit and 64-bit keys.\bigskip
\end{abstract}

\newpage 
\section{Introduction}
\paragraph{Independent hashing} The concept of $k$-independent hashing  
was introduced by Wegman and
Carter~\cite{wegman81kwise} at FOCS'79 and has been the cornerstone of
our understanding of hash functions ever since. The hash functions map
keys from some universe $U$ to some range $R$ of hash values.
Formally, a family $\calH = \{ h \mid  U \to R \}$ of hash functions is
$k$-independent if (1) for any distinct keys $x_1, \dots, x_k \in U$,
the hash values $h(x_1), \dots, h(x_k)$ are independent random
variables when $h$ is picked at random from $\calH$; and (2) for any fixed $x$, $h(x)$ is uniformly distributed
in $R$. By a $k$-independent hash function we refer to a function
chosen at random from such a family. Often the family is only given
implicitly as all possible choices some random parameters 
defining the function.

As the concept of independence is fundamental to probabilistic
analysis, $k$-independent hash functions are both natural and powerful in
algorithm analysis. They allow us to replace the heuristic assumption
of truly random hash functions with real (implementable) hash
functions that are still ``independent enough'' to yield provable
performance guarantees. We are then left with the natural goal of
understanding the independence required by algorithms.
When first we have proved that $k$-independence suffices for a
hashing-based randomized algorithm, then we are free to use {\em
  any\/} $k$-independent hash function.  

Let $U$ and $R$ be the sets $U=[u]=\{0,\ldots,u-1\}$ and
$R=[r]=\{0,\ldots,r-1\}$. The canonical construction of a $k$-independent family
is a polynomial of degree $k-1$ over a prime field $\mathbb
Z_p$ where $p\geq u$.  The random parameters are the coefficients $a_0, \dots, a_{k-1}
\in \mathbb Z_p$. The hash function is then
\begin{equation}\label{eq:polynomial}
 h(x) = \Big( \big( a_{k-1} x^{k-1} + \cdots + a_1 x + a_0 \big)
        \bmod{p} \Big) \bmod{r} 
\end{equation}
For $p\gg r$, the hash function is statistically close to $k$-independent. 
One thing that makes polynomial hashing over $\mathbb Z_p$ 
slow for $\wp> 2^{32}$ is that each multiplication over $\mathbb Z_p$ translates into
multiple 64-bit multiplications that due to discarded overflow can
only do exact multiplication of 32-bit numbers. The ``mod $p$''
operation is very expensive in general, but \cite{carter77universal}
suggests using a Mersenne prime $p$ such as $2^{61}-1$ or $2^{89}-1$,
and then 'mod $p$'  can be made very fast.

\paragraph{Word RAM model} We are assuming the word RAM model where
the operations on words are those available in a standard programming
language such as C~\cite{KR88}. A word defines the maximal
unit we can operate on in constant time. For simplicity, we 
assume that each key or hash value fits in a single  word. This implies that the time it takes to evaluate the degree $k-1$ polynomial from \req{eq:polynomial} is $O(k)$.
The Random Access Memory (RAM) implies that we can create tables, accessing
entries in constant time based on indices
computed from key values. Such random
access memory has been assumed for hash tables since Dumey introduced
them in 1956 \cite{Dum56}.

\paragraph{Time-space trade-offs}
To get faster hash functions, we implement them in two phases. First
we have a preprocessing phase where we based on a random seed
construct a representation of the hash function. We do not worry too
much about the resources used constructing the representation, but we
do worry about the space of the representation, measured in number of
words.  Next we have a typically deterministic query phase where we
for a given key compute the hash value using the representation. Table
\ref{tab:hashing} presents an overview of the results in this model
that will be discussed here in the introduction. In our $O$-notation,
we view both $u$ and $c$ as going to infinity, but $c$ is much smaller
than $u$.

\begin{table}
\begin{center}
\begin{tabular}{|c|c|c|l|}\hline
\multicolumn{4}{|c|}{Non-Constructive Cell-Probe Model}\\\hline
Space &  Probes & Independence  & Reference\\\hline
$u^{1/c}$ & - & $\leq u^{1/c}$    & Trivial\\
$u^{1/c}$ & $t<c$ & $\leq t$    &\cite{siegel04hash}\\
$u^{1/c}$ & $O(c)$ & $u^{\Omega(1/c)}$    &\cite{siegel04hash}\\[1ex]\hline
\multicolumn{4}{|c|}{C-programmable Word RAM model}\\\hline
Space & Time & Independence  & Reference \\\hline
$k$ & $O(k)$ & $k$    & Polynomial \\
$u$ & $1$ & $u$    & Complete table \\
$u^{1/c}$ & $O(c)^{\,c}$ & $u^{\Omega(1/c^2)}$    &\cite{siegel04hash}\\
$u^{1/c}$ & $O(ck)$ & $k$    &  \cite{dietzfel03tabhash,KW12,thorup12kwise}\\
$u^{1/c}$ & $O(c)$ & $u^{\Omega(1/c^2)}$    & This paper\\
$u^{1/c}$ & $O(c^{\,\lg c})$ & $u^{\Omega(1/c)}$  & This paper\\\hline
\end{tabular}
\end{center}
\caption{Hashing with preprocessed representation.}\label{tab:hashing}
\end{table}

In the case of polynomial hashing, the preprocessing
just stores the coefficients $a_0,....a_k$ in $k$ words. Unfortunately,
to find the hash value of a key $x$, we have to access all $k$ words in $O(k)$ time.
Another extreme would be to store the hash values of all possible keys
in one big table of size $u$. Then we can find the hash of a key in constant time by a single lookup. 

There has been interesting work on representing a high degree
polynomial for fast evaluation \cite[Theorem 5.1]{KU11}. For a degree
$k-1$ polynomial over $\mathbb Z_p$, the evaluation time is $(\log
k)^{O(1)}(\log p)^{1+o(1)}$. This avoids the linear dependence on $k$,
but the factor $\log p\geq \log u$ is prohibitive for our purposes.

\paragraph{Simple tabulation hashing}
In simple tabulation hashing, we generally view both keys and hash values
as bit strings, so $u$ and $r$ are powers of two. Moreover, we view
a key $x$ as a vector of $c$ characters $x_0, \dots, x_{c-1}$ from the
alphabet $\Ain=[u^{1/c}]$. {\em Simple tabulation\/} is defined
in terms of $c$ {\em character tables\/} $h_0,\ldots,h_{c-1}:\Ain\fct R$.
 This induces a function $h:U\fct R$
defined by 
\begin{equation}\label{eq:simple-table}
h(x) = \bigoplus_{i\in [c]} h_i(x_i)=h_0(x_0) \oplus \dots \oplus h_{c-1}(x_{c-1}).
\end{equation}
Here $\oplus$ denotes bit-wise exclusive-or (xor). We call this {\em simple tabulation hashing\/}  when the character tables
are filled with random hash values from $R$.
This is a well-known
scheme dating back at least to Zobrist in 1970
\cite{zobrist70hashing} who used it for game playing programs. 
Simple tabulation hashing is only
$3$-independent even if all character tables are fully random. 

In simple tabulation, the preprocessing phase
fills the $c$ character tables $h_i$. These may all be stored
consecutively as a single 2D array $[c]\times\Phi\fct R$ using
$c u^{1/c}$ space. If we already
have some randomly filled memory, then
a simple tabulation hash function is defined in constant time, 
simply by placing the offset of the array in the random 
memory. 

In the query phase, we find each $h_i(x_i)$ by a single lookup. We do
only $c$ lookups, and we only have a constant number of word
operations per lookup, so each hash value is computed in $O(c)$ time.
If $\Ain$ consists of 8-bit or 16-bit characters, then the character
tables fit in fast cache. For 32-bit or 64-bit
keys, simple tabulation is about 3 times faster than the 3-independent
hashing obtained as in \req{eq:polynomial} by a degree 2-polynomial
tuned for a Mersenne prime (see, e.g., experiments in
\cite{patrascu12charhash,thorup12kwise}). Also note that with 
simple tabulation, the cost of expanding the range $R$ to longer bit-strings is minor in that we still only have to
do $c$ lookups. The added cost is only from storing and looking up longer
bit-strings that have to be xor'd.

In \cite{patrascu12charhash} it was proved for many concrete applications
that simple tabulation has far more power than its 3-independence suggests.
However, to use simple tabulation in an application such as 
linear probing, one has
to make a careful analysis to show that the dependence is not harmful
to the application. This is not as attractive as the generic independence
paradigm where any $k$-independent hash function can be used in any application for which $k$-independence suffices. According to Google Scholar, Siegel's
\cite{siegel04hash} highly independent hashing has more than 150 citations 
(including those to the original conference version), but as he states,
it is  ``far too slow for any practical application''. 

\subsection{Results} In this paper we show that to get the same 
high independence as Siegel \cite{siegel04hash} efficiently, we just have to apply
simple tabulation twice, and we get even higher independence with more
applications. Our key is to show that simple tabulation,
applied once, is likely to have some strong expander properties.

\paragraph{Unbalanced expanders by simple tabulation}
To describe the result, we need some simple notation and terminology. Suppose $y\in \Aout^d$
is a vector of $d$ characters from $\Aout$.  We let $y_j$ denote
character $j$ in $y$, so $y=(y_0,\ldots,y_{d-1})$.
By a {\em position character\/} we mean a pair $(j,a)\in [d]\times \Aout$
consisting of a position and a character. The vector $y$ is
identified with the corresponding set $\{(j,y_j) | j\in [d]\}$ of position
characters.

Consider a function $f:U\fct \Aout^d$. It defines an unbalanced
bipartite graph with the key set $U$ on the left-hand side and the
output position characters from $V=[d]\times \Aout$ on the right-hand side. A
key $x\in U$ has $d$ distinct neighbors; namely the $d$ output
position characters $(0,f(x)_0),\ldots (d-1,f(x)_{d-1})\in V$. Two keys $x$
and $y$ share a neighboring output position character if and only if
$f(x)_j=f(y)_j$ for some $j$. We say a set $X\subseteq U$ has a {\em unique output
  position character $(j,a)$\/} if there is an $x\in X$ such that
$f(x)_j=a$ and for all other $y\in X\setminus \{x\}$, $f(y)_j\neq
a$. Our basic result is that 
if we consider random simple tabulation 
with 6 times more output than input characters, then every
not too large set $X$ has a unique output position character. This can
be viewed as a weak expander property. As the number of output
characters increases, we get the standard expansion property that
$X$ has $\Omega(d|X|)$ distinct output position
characters (neighbors in the bipartite graph). The formal statement is as follows.
\begin{theorem}\label{thm:main}
Consider a simple tabulation function $h:\Ain^c\fct \Aout^d$ where
$d\geq 6c$ and where the character tables are fully random. Assume
$c=|\Ain|^{o(1)}$ and $(c+d)^c=|\Aout|^{o(1)}$.  Let
$k=|\Aout|^{1/(5c)}$. With probability
$1-o(|\Ain|^2/|\Aout|^{d/(2c)})$,
\begin{itemize}
\item[(a)] every key set $X\subseteq \Ain^c$ of size 
$|X|\leq k$ has at least one unique output position character.
\end{itemize}
Moreover, for any $\eps\in(0,1)$, with
probability $1-o(|\Ain|^2/|\Aout|^{\eps d/(2c)})$,
\begin{itemize}
\item[(b)] every key set $X\subseteq \Ain^c$ of size 
$|X|\leq k$ has more than $(1-\eps)d|X|$ distinct output position characters.
\end{itemize}
The requirement that the character tables are fully random can be relaxed
in the sense that we for (a) and (b) can use any $k\leq |\Aout|^{1/(5c)}$ such
that all character tables are $k$-independent, and independent of each other.
\end{theorem}
Above we think of $c$ and $d$ as slow growing. Our construction
is interesting also when $d$ is constant, but then we cannot
measure its effect with $O$-notation. 

The assumptions $c=|\Ain|^{o(1)}$ and $(c+d)^c=|\Aout|^{o(1)}$ are
not essential, but serve
to give simple probability bounds for (a) and (b). As we shall see in
Theorem\tref{thm:practice}, we can derive much better bounds for
concrete cases.

Our work is orthogonal to the deep work on explicit expanders; for
Theorem\tref{thm:main} relies on random values for the character
tables.  Also, when it comes to highly unbalanced expanders like in
Theorem~\ref{thm:main}, the best explicit constructions \cite{GUV09}
have logarithmic degrees. It would be very interesting if we could
fill the character tables of Theorem~\ref{thm:main} 
explicitly during preprocessing with an
efficient deterministic algorithm. When done, we would enjoy the high
speed of simple tabulation.

\drop{Note how the error probability drops exponentially
with larger $d$. Informally, for $|\Phi|\leq |\Psi|$ and
$\eps=\Omega(1)$, we will say that (a) and (b) are satisfied {\em with
  high probability\/} with large enough $d=O(c)$.}

\paragraph{High independence by double tabulation}
In this paper, we are mostly interested in the unique output position
characters from (a). We say that a function $f:U\fct \A^d$ is {\em $k$-unique,} or has
{\em uniqueness $k$,} 
if every subset $X\subseteq U$ of size at most $k$ has a 
unique output position character. Translating Lemma 2.6 in \cite{siegel04hash},
we get
\begin{lemma}[Siegel]\label{lem:siegel} Let $f:U\fct \Aout^d$ be a $k$-unique function.
Consider a random simple tabulation function $r:\Aout^d\fct R$ 
where the character tables $r_j:\Aout\fct R$, $j\in[d]$, are 
independent of each other, and where each $r_j$ is $k$-independent. Then $r\circ f:U\rightarrow R$ is $k$-independent. 
\end{lemma}
For completeness, we include the proof of Lemma \ref{lem:siegel} 
in Appendix \ref{app:siegel}.

Suppose we have a concrete simple tabulation function $h:\Ain^c\fct
\Aout^d$ that satisfies (a) from Theorem
\ref{thm:main}. Then $h$ is $k$-unique. 
We can now compose $h$
with a random simple tabulation function $r:\Aout^d\fct R$ from Lemma
\ref{lem:siegel}. The resulting function $r\circ h$ is a $k$-independent function from $U=\Ain^c$ to $R$. We call this composition {\em double
  tabulation}.

Note that if we want a new independent $k$-independent 
hash function, we can still use the same $k$-unique $h$ as a universal constant. We only need to generate a new independent 
simple tabulation hash function
$r':\Aout^d\fct R$, and use $r'\circ h:U\rightarrow R$ as the new $k$-independent
hash function.

Unfortunately, we do not know of any efficient way of testing if the
simple tabulation function $h$ from Theorem\tref{thm:main} is
$k$-unique. However, a random $h$ is $k$-unique with
some good probability. To emphasize that we only need
$k$-uniqueness for a single universal $h:\Ain^c\fct
\Aout^d$, we say that it happens with
{\em universal probability}. 

\begin{corollary}\label{cor:main}
Let $u=|U|$ and assume $c^{\,c^2}=u^{o(1)}$.
With universal probability $1-o(1/u^{1/c})$, using space $o(u^{1/c})$, 
we get $u^{\Omega(1/c^2)}$-independent hashing from $U$ to $R$ in $O(c)$ time.
\end{corollary}
\begin{proof}
We use the above double tabulation. 
For simplicity, we assume that $u$ is a power of a power of two. 
For the first simple tabulation function $h$
from Theorem\tref{thm:main}, 
we use $c'=2^{\ceil{\lg_2 c}+1}$ input
characters from $\Ain$ and $d=8c'$ output characters, also from $\Ain$.
The uniqueness we get is
$k=|\Ain|^{1/(5c')}=|\Ain|^{\Omega(1/c)}$, and the error probability
is $o((1/u^{1/c'})^{2-d/(2c')})=o(1/u^{1/c})$.
The second simple tabulation
$r$ from Lemma \ref{lem:siegel} has $d$ input characters
from $\Ain$, so the total number of tables is $c'+d=O(c)$.
This is also the number of lookups, and for each lookup, we do a constant number of operations on a constant number of words. 
The space is thus $O(c u^{1/c'})=o(u^{1/c})$, and the evaluation time is $O(c)$.
\end{proof}
Siegel \cite{siegel04hash} has proved that with
space $u^{1/c}$ one needs evaluation time $\Omega(c)$ to get independence
above $c$. The time bound in Corollary \ref{cor:main} is thus optimal 
for any higher independence. We note that the restriction $c^{\,c^2}=u^{o(1)}$
is equivalent to saying that the independence $u^{\Omega(1/c^2)}$ is
more than polynomial in $c$.


\paragraph{Higher independence by recursive tabulation}
With representation space $u^{1/c}$, the highest independence we can
hope for is $u^{1/c}$. In Corollary \ref{cor:main}
we only get independence $u^{\Omega(1/c^2)}$. 
We will show that
we can get independence  $u^{\Omega(1/c)}$ using recursive tabulation.
This is where it is important that Theorem\tref{thm:main} allows
different alphabets for input and output characters. The basic
idea is to use output characters from $\Aout=[u^{1/2}]$, and recurse
on them to prove:
\begin{theorem}\label{thm:recurse}
Let $u=|U|$ and assume $c^{\,c^2}=u^{o(1)}$.
With universal probability $1-o(1/u^{1/c})$, using space $o(u^{1/c})$, we can get
$u^{\Omega(1/c)}$-independent hashing from $U$ to $R$ in $o(c^{\lg_2 c})$ time.
\end{theorem}
If we unravel the recursion (to be presented in Section
\ref{sec:recurse}), for some $D=o(c^{\lg c})$ and $k=u^{\Omega(1/c)}$,
we get a function $f:U\fct [u^{1/(2c)}]^D$ that is not $k$-unique, yet
which yields $k$-independence if composed with a random simple
tabulation function $r:[u^{1/(2c)}]^D\fct R$.
If follows from \cite[Proposition 2]{thorup12kwise} or \cite[Theorem 3]{KW12}
that $f$ has the property that some output position character appears
an odd number of times.

\paragraph{Concrete parameters} Note that when dealing with $n$ keys, it
is fairly standard to use {\em universe reduction}, applying universal hashing into a domain of
size $n^{2+\eps}$, $\eps=\Omega(1)$, hoping for no collisions. Starting from this domain,
dividing into $c=3$ characters brings us down to space $O(n^{2/3+\eps})$ which
may be very acceptable. Thus it is often reasonable to think of $c$ as small.

Below we consider some concrete parameter choices yielding
100-independent hashing. This would have been
prohibitively slow with the polynomial hashing from \req{eq:polynomial}.
With reference to Lemma \ref{lem:siegel}, the challenge is
to find a 100-unique function. The probabilities 
are based on careful calculations yielding much better bounds than
those derived from the simple formula in Theorem~\ref{thm:main}~(a).
We do not make any assumptions like $c=|\Ain|^{o(1)}$ and $(c+d)^c=|\Aout|^{o(1)}$.
\begin{theorem}\label{thm:practice} We consider a simple tabulation
hash function $h:\Ain^c\fct \Aout^d$. Assuming that the 
character tables $h_i$ of $h$ are fully random, or
at least 100-independent, and independent of each other,
\begin{enumerate}
\item\label{32-2} For 32-bit keys with $\Ain=\Aout=[2^{16}]$, $c=2$, and $d=20$, 
the probability that $h$ is not 100-unique is bounded by 
$1.5\times 10^{-42}$. 
\item\label{64-3}  For 64-bit keys with $\Ain=\Aout=[2^{22}]$, $c=3$, and $d=24$, 
the probability that $h$ is not 100-unique is bounded by 
$1.4\times 10^{-49}$. 
\item\label{64-4} For 64-bit keys with $\Ain=[2^{16}]$, $\Aout=[2^{32}]$, $c=4$, and 
$d=14$, 
the probability that $h$ is not 100-unique is bounded by 
$9.0\times 10^{-36}$. The idea is to use triple tabulation, applying Case \ref{32-2} to each of
the 32-bit output characters. 
\end{enumerate}
\end{theorem}
Recall that we only need a single universal 100-unique function $h$ for
each set of parameters. Trusting some randomly filled memory 
to represent such a 100-unique function as in Theorem~\ref{thm:practice} is extremely safe.

\drop{
The actual speeds depends on the size and speed of the
available memory. The expensive part is not the $c$ lookups by the
above 100-unique $h$. The expensive part is the composition with
the simple tabulation hashing $r$ from Lemma \ref{lem:siegel}
where we have to do a lookup in each of the $d\geq c$ character tables $r_j$, $j\in d$. If the lookups are in cache, they are
still fast compared with multiplication, and modern computers have
megabytes of cache. Moreover, all $r$ is an array offset in some
randomly filled memory. 
We also note that tabulation hashing defined by offsets
to some given randomly filled memory The above tabulation hash functions can be shared efficiently in
multi-core systems. The point is that the randomly filled memory used
for the tables does not need to be changed. In a standard multi-core
protocols such as
MESI \footnote{\texttt{http://en.wikipedia.org/wiki/MESI\_protocol}},
the memory thus remains clean, and it is therefore shared freely
between caches. We could even imagine future computers with some randomly filled read-only memory. This should be both cheaper and faster than regular memory, which could be saved for other purposes.}

\subsection{Siegel's highly independent hashing}
Siegel's study on hashing~\cite{siegel04hash} considered the fundamental trade-offs between 
independence, representation space, and the time it takes to
compute the hash of a key.
\paragraph{Lower bound}
Siegel's lower bound
\cite[Theorem 3.1]{siegel04hash} is in Yao's \cite{yao81tables}
powerful cell probe model. To get clean bounds, he assumes that
the domain of a word or cell is no bigger that of a single hash value.
Trivially this means that we need at least $k$ cells to get independence $k$.

The representation is an arbitrary function of the random seed. If
the representation has $s$ cells, an equivalent formulation is that
the contents of the $s$ cells follow an arbitrary distribution.

The querier is given the key. To compute the hash value, he can probe
the cells of the representation. He is only charged for these cell probes. 
His next move is an arbitrary function of the key and the cells he has read so far: he can either
pick a cell based on this information, or output the hash value.
 
Siegel shows that if the representation uses $u^{1/c}$ cells, and the
query phase makes $t<c$ probes, then the hash function computed can be
at most $t$-independent. His argument is very robust, e.g., with no
change to the asymptotics, he can allow some quite substantial bias in
the independence, look at average query time, etc.

\paragraph{Upper bounds}
Siegel's framework for upper bounds is similar to what we already
described, but simpler and in that he is not ``position sensitive'': Given a function $f:U\fct \Aout^d$, he considers the
unbalanced bipartite graph with the keys from $U$ on the left-hand side, and 
output characters from $\Aout$ on the
right-hand side (on our right-hand side, we had the position output
characters from $V=[d]\times\Aout$).  A key $x\in U$ has the $d$ neighbors
$f(x)_0,\ldots,f(x)_{d-1}$ that may not all be distinct. He says that $f$ is
$k$-peelable (corresponding to $k$-unique) if every key set $X$ of
size at most $k$ has a unique output character. 
Here $x,y\in X$ share an output character if $f(x)_i=f(y)_j$ even
if $i\neq j$. He uses a {\em single\/}
character table $r_0:\Aout\fct R$, and defines
$r:\Aout^d\fct R$ by 
\begin{equation}\label{eq:single-table}
r(x)=\bigoplus_{j\in [d]}r_0(x_j).
\end{equation}
Siegel proves \cite[Lemma
  2.6]{siegel04hash} that if $f$ is $k$-peelable, and
$r_0:\Aout\fct R$ is random, then $r\circ f$ is $k$-independent. Note that the space of $r$ is independent
of $d$ since $r_0$ uses only a single character table taking space $|\Aout|$. 
It
does, however, take $d$ lookups to evaluate \req{eq:single-table}. The
problem is to find the $k$-peelable function $f$.

Let $u=|U|$ and $u^{1/c}=|\Aout|$. For the existence of a $k$-peelable function, Siegel
\cite[Lemma 2.9]{siegel04hash} argues
that a fully random $f:U\fct \Aout^d$ is likely to be a good expander 
from $U$ to $\Aout$ if $d\geq 6c$. More precisely, with probability
$1-\tO(1/u)$, for
$k=u^{1/(2c)}$, he gets that every set $X$ of size $|X|\leq k$ has
more than $d|X|/2$ neighbors. He also notes \cite[Lemma 2.8]{siegel04hash} 
that if $X$ has more than
$d|X|/2$ distinct neighbors, then some of them have to be unique, so $f$ is also $k$-peelable. 

Representing a fully random $f$ would take space $u$, but existence is
all that is needed for upper bounds in the abstract cell-probe model.
We can simply use the unique lexicographically smallest
$k$-peelable $F=\min\{f:U\fct\Aout^d\mid 
f\textnormal{ is $k$-peelable}\}$. The 
querier can identify $F$ on-the-fly without any probes.
The representation only needs to
include the random $r_0$ which takes $u^{1/c}$ space.  The hash
$r(F(x))$ of a key $x$ is computed with $d=O(c)$ probes to $r_0$, and
the independence is $k=u^{1/(2c)}$. The number of probes is within a
constant factor of the lower bound which says that with $u^{1/c}$, we need
at least $c$ probes for any independence above $c$.

To get an implementation on the word RAM \cite[\S
  2.2]{siegel04hash}, Siegel makes a graph product based
on a small random graph that can be stored in space $u^{1/c}$.
Assuming that the random graph has sufficient expander properties,
the product 
induces a $u^{\Omega(1/c^2)}$-peelable function
$f:U\fct\Aout^{O(c)^{\,c}}$.  
This leads to a
$u^{\Omega(1/c^2)}$-independent hash function represented in $u^{1/c}$
space. Hash values are computed in $O(c)^{\,c}$ time. It should be
noted that Siegel's focus was the case where $c=O(1)$, and then he does
get $u^{\Omega(1)}$-independence in $O(1)$ time, but here we consider
$c=\omega(1)$ in order to qualify the dependence on $c$.

The RAM implementation of Siegel should be compared with our bounds
from Theorem\tref{thm:recurse}: $u^{\Omega(1/c)}$-independent hashing
using $o(u^{1/c})$ space, computing hash values in $o(c^{\lg c})$
time. Our independence is significantly higher---essentially as high
as in his existential cell-probe construction---and we are almost
exponentially faster. We should also compare with Corollary
\ref{cor:main}: $u^{\Omega(1/c^2)}$-independent hashing using
$o(u^{1/c})$ space, computing hash values in $o(c)$ time. This is the same
independence as an Siegel's RAM implementation, but with the optimal
speed of his existential cell probe construction.

On the technical side, recall that Siegel's $k$-peelability is not
position sensitive. This is only a minor technical issue, but being
sensitive to positions does yield some extra structure. In particular,
we do not expect the simple tabulation function from Theorem\tref{thm:main} to
be $k$-peelable without the positions.

\subsection{Other related work}
Siegel states \cite[Abstract]{siegel04hash} about his scheme that it
is ``far too slow for any practical application''. This and the
$O(c)^{\,c}$ evaluation time has lead researchers to seek simpler and
faster schemes. Several works
\cite{dietzfel03tabhash,KW12,thorup12kwise} have been focused on the
case of smaller independence $k$. These works have all been
position sensitive like ours.  Fix $\Aout=[u^{1/c}]$. We are looking
at functions $f:U\fct \Aout^d$, to be composed with a simple tabulation
hash function $r:\Aout^d\fct R$. The evaluation time is $O(d)$, so we
want $d$ to be small.

Dietzfelbinger and Woelfel \cite[\S 5]{dietzfel03tabhash} pick
$d$ 2-independent hash functions $f_0,\ldots,f_{d-1}:U\fct\Aout$. This
yields a function $f:U\fct \Aout^d$ defined by $f(x)=(f_0(x),\ldots,f_{d-1}(x))$.
Composing $f$ with a random simple tabulation function $h:\Aout^d\fct R$,
they show that the result is close to $k$-independent if $d\gg kc$.

Thorup and Zhang \cite{thorup12kwise} found an explicit deterministic construction of
a $k$-unique $f$ which also has better constants than the scheme from 
\cite{dietzfel03tabhash}. By Lemma \ref{lem:siegel}, the resulting hash 
function is exactly $k$-independent. Simple tabulation is by
itself $3$-independent, but  \cite{thorup12kwise} is motivated
by applications needing 4 and 5-independence.
For $k=5$ and $\Aout=[u^{1/c}+1]$,
\cite{thorup12kwise}  gets down to $d=2c-1$. For general $k$, using
$\Aout=[u^{1/c}]$, \cite{thorup12kwise} gets $d=(k-1)(c-1)+1$.

Klassen and Woelfel \cite{KW12} focus mostly on $c=2$, where 
for arbitrary $k$ they get $d=(k+1)/2$. For general $c$, their bound
is $d=\ceil{2\frac{c-1}{2c-1}(k-1)}(c-1)+1$.

We note that the twisted tabulation in \cite{PT13:twist} has a similar flavor to the above
schemes, but it does not yield independence above 3. The main target
of \cite{PT13:twist} is to get strong Chernoff style bounds.

The above works \cite{dietzfel03tabhash,KW12,thorup12kwise} thus need
$d=\Omega(kc)$ for independence $k$. This contrasts our Theorem\tref{thm:main}
which gets $d=O(c)$ with independence $u^{\Omega(1/c)}$. 
Apart from the case $c=2$, $k=5$ from \cite{thorup12kwise}, our new
scheme is probably also the easiest to implement, as we are just applying
simple tabulation twice with different parameters.

There are also constructions aimed at providing good randomness for a
single unknown set $S$ of size $n$ \cite{dietzfel09splitting,PP08}.
In particular, Pagh and Pagh \cite{PP08} have a two-part randomized
construction of a constant time hash function $h$ that uses $O(n)$
space so that for any given set $S$ of size $n$, if Part 1 does not
fail on $S$, then Part 2 makes $h$ fully random on $S$. We have the
same two-parts pattern in our double tabulation where Part 1 generates
a random simple tabulation function that we hope to be $k$-unique on the
whole universe, and Part 2 composes this function with another random simple tabulation
function $r$. If Part~1 succeeds, the result is $k$-independent. A
principal difference is that any concrete fixing of Part 1 from
\cite{PP08} fails for many sets $S$, so the success probability of
Part 1 in \cite{PP08} is not universal; otherwise this would have been
an $n$-independent hash function. From a more practical perspective,
often we only need, say, $\log n$-independence, and then double
tabulation with universe reduction and 
small character tables in cache is much simpler and faster than  \cite{PP08}. In fact, \cite{PP08} uses Siegel's
\cite{siegel04hash} highly independent hash functions as a subroutine,
and now we can instead use our double tabulation. Double tabulation
fits very nicely with the other use of random tables in \cite{PP08}, making
the whole construction of full randomness for a given set $S$ quite simple. It should be noted that
\cite{dietzfel03tabhash} have found a way of bypassing the need of
\cite{siegel04hash} in \cite{PP08}.  However, our double tabulation is
even simpler, and it replaces the use of \cite{siegel04hash} in all
applications.

\section{The basic analysis}\label{sec:basic}
The next two sections are devoted to the proof of Theorem\tref{thm:main}. For now, we assume that all character tables are fully random, leaving the relaxation to $k$-independent character tables till the very end.

By an \emph{input position
  character} we mean a value from $[c]\times \Ain$. Notationally, we can 
then
view a key $x=(x_0,\ldots,x_{c-1})$ as the set of input position characters:
$\{(0,x_0),\ldots,(c-1,x_{c-1})\}$.
We can now specify $\hi $ as a single table from input position characters
$[c]\times\Ain$ to vectors $\hi (\alpha)\in \Aout^d$, that
is, if $(a,i)=\alpha\in[c]\times\Ain$, then $\hi(\alpha)=h_i[a]$.
This view induces a function $\hi $ on arbitrary sets $x$ of input position characters:
\begin{equation}\label{eq:set-def}
\hi (x) = \bigoplus_{\alpha \in x} \hi (\alpha).
\end{equation}
Note that when $x$ is the set corresponding to a key, \req{eq:set-def}
agrees with \req{eq:simple-table}. We define an {\em output index\/} as a pair 
$(\alpha,j)\in ([c]\times \Ain)\times [j]$ indexing the individual
output character $h(\alpha)_j$.

We want to show that the if we assign $h:[c]\times\Ain\fct \Aout^d$ at
random, then there is only a small probability that there exists
a set $X\subseteq \Ain^c$, $|X|\leq k\leq |\Aout|^{1/(5c)}$, violating (a) or (b) in Theorem\tref{thm:main}. 

\paragraph{Efficient coding}
To specify $\hi$, we have to specify a vector of  $d$ output 
characters from
$\Aout$ for each of the $c|\Ain|$ input position characters.  Based on a
violating set $X$, we will construct an efficient coding of some of the
output characters. The number of such efficient codings will be
much smaller than the number of ways we can assign the output characters
coded. Efficient codings are therefore rarely possible, hence so are
the violating sets.

Our coding will not describe the set $X$, and it is important that
decoding can be done without any knowledge of $X$, except that $|X|\leq k$.
The coding starts by specifying a list $L$ with some of the input position characters from the
keys in $X$. We will now go through the input position characters $\alpha\in L$ 
in the order that they appear in $L$. For each $\alpha$, we
will specify the $d$ output characters $\hi (\alpha)_j$, $j\in[d]$. 
Some of these output characters will be derivable from
previously specified output characters, leading to a more efficient encoding:

\begin{definition} We say the output character $h(\alpha)_j$ is \emph{derivable} if
 there exist keys $x,y\in X$ such that:
\begin{itemize*}
\item The symmetric difference $x\symmdiff y=\{(i,x_i),(i,y_i)\, |\, i\in [c], x_i\neq y_i\}$ of $x$ and $y$ is contained in $L$.
\item $\alpha$ is last in $L$ among the input position characters 
in $x\symmdiff y$.
\item $h(x)_j=h(y)_j$, or equivalently, $h(x\symmdiff y)_j = 0$.
\end{itemize*}
\end{definition}
In our representation, we do not need to know the keys $x$
and $y$. We only need to know the symmetric difference $A=x\symmdiff y\subseteq
L$. We call $(A,j)$ an {\em equation\/} as it
represents the information that $h(A)_j=0$. 
The output index specified by the equation $(A,j)$ is
the pair $(\alpha,j)$ where $\alpha$ is the last input position character from
$A$ in the list $L$. The equation derives the output character
\[h_j(\alpha)=\bigoplus\{h(\beta)_j\mid \beta\in A \setminus \{\alpha\}\}.\] 
The input position characters $\beta$ all
precede $\alpha$ in $L$, so the output characters $h(\beta)_j$ have
all been specified. We do not want more than one equation specifying the
same output index.

When the list $L$ of length $\ell$ is given,
the set $A$ can be picked in less than
$\ell^{2c}$ ways, so the number of possible
derivations is less than $\ell^{2c}d$.
If $\ell^{2c}d\ll |\Aout|$,
then this is a win.  Indeed this is the case because 
$\ell\leq kc\leq c|\Aout|^{1/(5c)}$ and $(c+d)^c=|\Aout|^{o(1)}$.
However, we will have to make a lot
of derivations to make up for the fact that we first have to specify
the $\ell$ input position characters in $L$. In Section \ref{sec:findingL}
we will show that a violating set $X$ implies the existence of
a list $L$ with many derivable output characters, e.g., a
violation of (a) in Theorem\tref{thm:main} will yield $|L|\,d/(2c)$ derivable
output characters.

Below, for a given parameter $q$, we study the probability $\Pq$ of
finding a list $L$ of length at most $kc$ with at least $q|L|$
derivable output characters.  Below we will prove that
\begin{equation}\label{eq:main}
\Pq=o(|\Ain|^2/|\Aout|^q). 
\end{equation}
There may be much more than $q|L|$ output characters derivable from $L$. However, in our encoding, we also only store equations for exactly
$\ceil {q|L|}$ of them.

\paragraph{Coding and decoding} To summarize, the exact components of
our code are:
\begin{enumerate}
\item A list $L$ of $\ell$ input position characters.
\item A set of $M$ of $\mq$ equations $(A,j)$ where $A\subseteq L$
and $j\in [d]$. Let $I$ be the set of
output indices specified in these equations. The output indices should
all be distinct, so $|I|=\mq$.
\item A reduced table $H$ that for each $(\alpha,j)\in (A\times
  [d])\setminus I$, specifies the output character $\hi (\alpha)_j\in
  \Aout$.
\end{enumerate}
Above, each component presumes that the previous components are known,
so $L$ is known when we specify $M$, and $L$, $M$, and hence $I$ is
known when we specify $H$. Together, this specifies $L$ and
$\hi |L$. The decoding of $\hi |L$ goes as follows.  From $L$ and $M$ we
compute the set $I$ of output indices $(\alpha,j)$ specified by $M$.
For all other output indices $(\alpha,j)\in L\times [d]$, we find the
output character $\hi (\alpha)_j$ in $H$. To get the
remaining output characters we run through the input position characters $\alpha\in L$
in the order they appear in $L$. For each
$\alpha$ and $j\in[d]$, we check if $(\alpha,j)\in I$. If so, we take the
corresponding equation $(A,j)\in M$, and set
$\hi (\alpha)_j=\hi_j (A\setminus\{\alpha\})$.

\paragraph{Bounding the probabilities}
Let the above coding be fixed, and consider a random simple tabulation
function $\hi $. The
probability that our coding matches $\hi (\alpha)_j$ for all output
indices $(\alpha,j)\in L\times [d]$ is exactly $1/|\Aout|^{\ell
  d}$. A union bound over all possible codes will imply that none of
them are likely to match a random $\hi $.

Let us first assume that $\ell$ is fixed, that is, we restrict
our attention to codes where $|L|=\ell$. The number of
choices for $L$ is bounded as $\choicesl(L)<(c|\Ain|)^\ell$. Let $\choicesql(M)$ 
be the number of choices for $M$ given $L$. We already saw that
the number of possible equations is bounded by $\ell^{2c}d$. The number of
ways we can pick $\mq$ of these is trivially bounded as
\begin{equation*}\label{eq:choiceR0}
\choicesql(M)< (\ell^{2c}d)^{\mq}.
\end{equation*}
Finally, we need to pick $H$ with an output character for
each output index in $(L\times [d])\setminus I$. There
are $\ell d-\mq$ output characters to pick, leaving us $|\Aout|^{\ell d-\mq}$
choices for $H$. All in all we have $\choicesl(L)\cdot\choicesql(M)\cdot
|\Aout|^{\ell d-\mq}$ possible codes with the given $\ell$. By
the union bound, the probability that any of them match a random $\hi $
is 
\begin{align}\label{eq:basic}
\Pql&=\frac{\choicesl(L)\cdot \choicesql(M)\cdot
|\Aout|^{\ell d-\mq}}{|\Aout|^{\ell d}}
=\frac{\choicesl(L)\cdot \choicesql(M)}{|\Aout|^{\mq}}\\
&<
(c|\Ain|)^\ell \left(\frac{\ell^{2c}d}{|\Aout|}\right)^\mq
\leq 
(c|\Ain|)^\ell \left(\frac{\ell^{2c}d}{|\Aout|}\right)^{q\ell}\label{eq:cheat}
\end{align}
Strictly speaking, the last inequality assumes 
$\frac{\ell^{2c}d}{|\Aout|}\leq 1$. However, if $\frac{\ell^{2c}d}{|\Aout|}> 1$,
the whole bound is above $1$, and hence a trivial upper bound
on $\Pql$. Since $\ell\leq ck\leq c|\Aout|^{1/(5c)}$,
we have
\[(c|\Ain|)^\ell \left(\frac{\ell^{2c}d}{|\Aout|}\right)^{q\ell}
\leq \left(|\Ain|/|\Aout|^{3q/5}c(c^{2c}d)^q\right)^\ell.\]
We will now use our assumptions $c=|\Ain|^{o(1)}$ and $(c+d)^c=|\Aout|^{o(1)}$.
We can also assume that $|\Ain|^2\leq |\Aout|^q$, for otherwise \req{eq:main}
is a trivial probability bound above 1. Hence $c=|\Ain|^{o(1)}=|\Aout|^{o(q)}$,
so $c(c^{2c}d)^q=|\Aout|^{o(q)}$. Hence
\[\Pql\leq 
\left(|\Ain|/|\Aout|^{(3/5-o(1))\,q}\right)^\ell.\]
However, we must have $\ell\geq 2$, for otherwise there cannot be any equations. Therefore
\[
\Pq\leq \sum_{\ell=2}^{ck} \Pql
=o(|\Ain|^2/|\Aout|^q).\]
This
completes the proof of \req{eq:main}.

Finally, as stated in Theorem\tref{thm:main}, we need to argue that we do not need the character tables of
$h$ to be fully random. For $k\leq|\Aout|^{1/(5c)}$, it should suffice
that the character tables are $k$-independents and independent of
each other. The simple point is that the violating set $X$ is of
size $|X|\leq k$, so it involves at most $k$ input characters for
each position, and $L$ can only use these input characters. With $k$-independent hashing, the assignment of output characters to the input characters in $L$
is completely random, so we
do not need any changes to the above analysis.

\section{Many derivable output characters.}\label{sec:findingL}
The goal of this section is to prove that if there is a set $X$ violating
(a) or (b) in Theorem  \ref{thm:main}, then we can construct a list $L$ with many
derivable characters. 
\begin{theorem}\label{thm:key-main} 
Consider a simple tabulation function $h:\Ain^c\fct \Aout^d$.
\begin{itemize}
\item[($\,\overline{a}$)] If there is a key set $X$ with no unique output position characters, then
there is a list $L$ with some of the input position characters from $X$ so 
that at least $\frac d{2c}|L|$
of the output characters from $L$ are derivable.
\item[($\,\overline{b}$)] If for some $\eps\leq 1$ there is a key set $X$ with at most
$(1-\eps)d|X|$ distinct output position characters, then
there is a list $L$ with some of the input position characters from $X$ so 
that at least $\frac{\eps d}{2c}|L|$
of the output characters from $L$ are derivable.
\end{itemize}
\end{theorem}
\newcommand\na{($\overline{\textnormal a}$)}
\newcommand\nb{($\overline{\textnormal b}$)} 
\paragraph{Proof that Theorem\tref{thm:key-main} implies Theorem
\ref{thm:main}} Before proving Theorem\tref{thm:key-main}, we note
that it trivially implies Theorem
\ref{thm:main}, for if there is a set $X$ violating Theorem\tref{thm:main}~(a), then $X$ satisfies Theorem\tref{thm:key-main}~\na{}, so there is a list $L$ with $\frac d{2c}|L|$ derivable characters. By \req{eq:main} the probability of this event is $P^{d/(2c)}\leq
|\Ain|^2/|\Aout|^{d/(2c)}$. Likewise Theorem\tref{thm:main}~(b) follows from
Theorem\tref{thm:key-main}~\nb{}. \qed

\begin{proof}[ of Theorem\tref{thm:key-main}] We assume that we have a set $X$ satisfying the conditions
of \na{} or \nb. For a uniform proof, if the condition of \na{} is true, 
we set $\eps=1$, overruling
a possibly smaller $\eps$ from \nb{}. Set $q=\frac{\eps d}{2c}$. We will identify the list $L$ so that at least $q|L|$ of the output characters
from $L$ are derivable.

Let $\alpha_1,\ldots,\alpha_{\ell^*}$ be the distinct
input position characters from keys in $X$ listed in order of
decreasing frequency in $X$. Let $n_i$ be the number of keys
from $X$ containing $\alpha_i$. Then $n_1\geq n_2\geq\cdots \geq n_{\ell^*}$ and $\sum_{i=1}^{\ell^*}n_i=c|X|$.

Let $L_{\leq\ell}$ be the prefix $\alpha_1,\ldots,\alpha_{\ell}$.
The list $L$ in the theorem will be $L_{\leq \ell}$ for some $\ell\leq \ell^*$.
Let $\cconf_\ell$ be the number of new derivable output characters when
$\alpha_\ell$ is added to $L_{\leq\ell-1}$ creating $L_{\leq \ell}$. Then
\[\cconf_\ell=|\{j\in[d]\ |\ \exists x,y\in X, \alpha_\ell \in x\symmdiff y\subseteq L_{\leq\ell}, h(x)_j=h(y)_j\}|\]
The list $L_{\leq\ell}$ satisfies the theorem if $\sum_{i=1}^\ell
\cconf_i\geq q\ell$.  To prove that this is true for some
$\ell\leq\ell^*$, we study a related measure
\[\Dconf_{\leq\ell}=|\{(x,j)\in X\times[d]\;|\; \exists\,y\in X\setminus\{x\}:x\symmdiff y\subseteq L_{\leq \ell},\; h(x)_j=h(y)_j\}|.\]
Then 
\[\Dconf_{\leq \ell^*}=|\{(x,j)\in X\times[d]\;|\;\exists\,y\in X\setminus\{x\}: h(x)_j=h(y)_j\}|\]
counts with multiplicity the number of non-unique output characters from $X$.
\begin{lemma}\label{lem:non-unique}
$\Dconf_{\leq \ell^*}\geq \eps d|X|$.
\end{lemma}
\begin{proof} Each key $x\in X$ has $d$ output position characters, so with
multiplicity, the total number of output position characters from $X$
is $d|X|$. In case \na{} these are all non-unique and we have $\eps=1$.

In case \nb{} we have at most $(1-\eps)d|X|$ distinct output characters
from $X$. The number of unique output position characters must be smaller,
so with multiplicity, the total number of non-unique output characters from $X$ is bigger than $\eps d|X|$. 
\end{proof}
The following key lemma relates the two measures:
\begin{lemma}\label{lem:DCA-ind}For $\ell=2,\ldots,\ell^*$,
\begin{equation}\label{eq:DCA-ind} 
\Dconf_{\leq\ell}-\Dconf_{\leq \ell-1}\leq 2 \cconf_\ell\, n_\ell.
\end{equation}
\end{lemma}
\begin{proof}
We want to count the pairs $(x,j)\in X\times[d]$ that are counted in $\Dconf_{\leq\ell}$
but not in $\Dconf_{\leq \ell-1}$. First we consider ``$\alpha_\ell$-pairs'' 
$(x,j)$ where $x$ contains $\alpha_\ell$ 
and there is a ``witnessing'' key $y\in X$ not containing $\alpha_\ell$ such that 
$x\symmdiff y\subseteq L_{\leq \ell}$ and
$h(x)_j=h(y)_j$. We note that in this case $(\alpha_\ell,j)$ is derivable,
so $j$ is counted in $\cconf_\ell$.  The number of $\alpha_\ell$-pairs $(x,j)$ is thus bounded by $\cconf_\ell\, n_\ell$.

With the above $x$ and $y$, we would also count the ``witnessing'' pair
$(y,j)$ if $(y,j)$ is not already counted in $\Dconf_{\leq \ell-1}$. Suppose
we have another pair $(z,j)$ witnessing $(x,j)$. Thus $x\symmdiff y,x\symmdiff y\subseteq L_{\leq \ell}$ and $h(x)_j=h(y)_j=h(z)_j$. 
We want to show that $z\symmdiff y\subseteq
L_{\leq \ell-1}$, hence that both $(y,j)$ and $(z,j)$ were already counted in $\Dconf_{\leq \ell-1}$.

All input position characters in $y\symmdiff z$ come in pairs $(i,y_i)$,
$(i,z_i)$, $y_i\neq z_i$. At least one of $y_i$ and $z_i$ is different
from $x_i$. By symmetry, assume $y_i\neq x_i$. Then $(i,y_i),(i,x_i)
\in y\symmdiff x\subseteq L_{\leq \ell}$. Therefore  $(i,z_i)\in L_{\leq \ell}$ if
$z_i=x_i$; but otherwise $z_i\neq x_i$ and $(i,z_i),(i,x_i) \in z\symmdiff
x\subseteq L_{\leq \ell}$. In either case, we conclude that $(i,y_i),(i,z_i)
\in L_{\leq \ell}$. But $\alpha_\ell$ is in neither $y$ nor $z$, so it follows
that $(i,y_i),(i,z_i) \in L_{\leq \ell-1}$, hence that $y\symmdiff z\subseteq L_{\ell-1}$.
We conclude that both $(y,j)$ and $(z,j)$ were counted in $\Dconf_{\leq \ell-1}$, or 
conversely, that we for each $\alpha_\ell$-pair $(x,i)$ have at most one
witnessing pair $(y,j)$ that is counted in $\Dconf_{\leq\ell}-\Dconf_{\leq \ell-1}$.

We conclude that the number of witnessing pairs is no bigger than the number of $\alpha_\ell$-pairs, hence that $\Dconf_{\leq\ell}-\Dconf_{\leq \ell-1}$ is at most $2\cconf_\ell\, n_\ell$. 
\end{proof}
By \req{eq:DCA-ind}, for $\ell=1,\ldots,\ell^*$, 
\begin{equation}\label{eq:DCA}
\Dconf_{\leq\ell}\leq 2 \sum_{i=1}^\ell \cconf_i\,n_i.
\end{equation}
Recall that $L_{\leq \ell}$ satisfies the statement of the theorem
if $\sum_{i=1}^\ell \cconf_i\geq q\ell$. Assume for a contradiction
that there is a $q'<q$, 
such that for all $\ell=1,\ldots,\ell^*$,
\begin{equation}\label{eq:restrict}
\sum_{i=1}^\ell \cconf_\ell\leq q'\ell.
\end{equation}
The $n_\ell$ are  decreasing, so the $\cconf_\ell$ values that 
satisfy \req{eq:restrict}
and maximize the sum in \req{eq:DCA} are all equal to $q'$. Thus
\req{eq:DCA} and \req{eq:restrict} implies that 
\[\Dconf_{\leq\ell}\leq 2 \sum_{i=1}^\ell \cconf_i\,n_i
\leq 2\sum_{i=1}^\ell  q' n_i< 2q\sum_{i=1}^\ell n_i.
\]
In particular, we get 
\begin{equation}\label{D-bound}
\Dconf_{\leq \ell^*}<2q\sum_{i=1}^{\ell^*} n_i=2q|X|c.
\end{equation}
Since $q=\eps d/(2c)$, this contradicts Lemma \ref{lem:non-unique}. Thus we 
conclude that there is an $\ell$ such that $L_{\leq \ell}$ satisfies the theorem. 
This completes the proof of Theorem\tref{thm:key-main}, hence of
Theorem\tref{thm:main}.
\end{proof}

\section{Higher independence with recursive tabulation}\label{sec:recurse}
We will now  use recursive tabulation to get
the higher independence promised in Theorem\tref{thm:recurse}:
\begin{quote}{\em
Let $u=|U|$ and $c^{\,c^2}=u^{o(1)}$, 
With universal probability $1-o(1/u^{1/c})$, using space $o(u^{1/c})$, we can get
$u^{\Omega(1/c)}$-independent hashing from $U$ to $R$ in $o(c^{\lg_2 c})$ time.
}\end{quote}
\begin{proof}[ of Theorem  {\ref{thm:recurse}}]
For simplicity, we assume that $u$ is a power of a power of two.
Let $\ell=\ceil{\lg_2 c}+1$,
$c'=2^\ell$, and $\Ain=[u^{1/c'}]$. The independence we aim for
is $k=u^{1/(10 c')}=u^{\Omega(1/c)}$.

Our construction is a recursion of depth $\ell$.
On level $i=0,...,\ell-1$ of the recursion, the input key universe is
$U_{(i)}=[u^{1/2^i}]$, and we want a $k$-independent
hash functions $U_{(i)}\fct R$.  The set of input characters
will always be $\Ain=[u^{1/c'}]$, so on level $i$, we have
$c_{(i)}=c'/2^{i}$ input characters.  We apply 
Theorem\tref{thm:main} with $d_{(i)}=12c_{(i)}$ output
characters from $\Aout_{(i)}=U_{(i+1)}$. With universal probability $1-|\Phi|^2/\Aout_{(i)}^6$, Theorem~\ref{thm:main} gives
us a simple tabulation function $h_{(i)}:\Ain^{c_{(i)}}\fct \Aout_{(i)}^{d_{(i)}}$ with uniqueness
\[|\Aout_{(i)}|^{1/(5c_{(i)})}=\left(u^{1/2^{i+1}}\right)^{1/(5(c'/2^{i}))}
\geq u^{1/(10 c')}=k\textnormal,\] as desired. To get $k$-independence
from $U_{(i)}$ to $R$, as in Lemma \ref{lem:siegel}, we compose
$h_{(i)}$ with a simple tabulation function
$r_{(i)}:\Aout_{(i)}^{d_{(i)}}\fct R$ where the character tabulation
functions $r_{(i),j}:\Aout_{(i)}\fct R$ have to by $k$-independent and independent of each other. Here
$\Aout_{(i)}=U_{(i+1)}$, and the $r_{(i),j}$ are constructed
recursively. At the last recursive level, the output characters
are from $\Aout_{(\ell-1)}=U_{(\ell)}=[u^{1/2^\ell}]=\Ain$. We will store an independent
random character table for each of these output characters.

On each recursive level $i< \ell$, we can use the
same universal $k$-unique simple tabulation function $h_{(i)}:\Ain^{c_{(i)}}\fct \Aout_{(i)}^{d_{(i)}}$. However, on
the bottom level, we need independent random character tables for all
the output characters. The total
number of output characters on the bottom level is
\[D=\prod_{i=0}^{\ell-1} d_{(i)}=\prod_{i=0}^{\ell-1} 12c'/2^i\leq
O(\sqrt{c})^{\lg_2 c}.\]
Handling all of these, on the bottom level, we
have a single large simple tabulation function
$r:\Ain^D\fct R$ where the $D$ character tables are fully random tables supporting look-ups in constant time.

On recursive level $i<\ell$, the size of the intermediate domain $\Aout_{(i)}^{d_{(i)}}$
is $\left(u^{1/2^{i+1}}\right)^{12c'/2^{i}}=u^{6c'}$. The
elements from this domain thus use $O(c)$ words. It follows that
the space used by $h_{(i)}$ is $O(c_{(i)}|\Ain|c)$, and that
its evaluation time from \req{eq:simple-table} is $O(c_{(i)}c)=O(c^2/2^i)$.

We only represent a single universal function $h_{(i)}$ on each level $i<\ell$, to
the total space is clearly dominated by the $D$ tables on the bottom
level.  The total space is therefore $O(D|\Ain])=O(\sqrt c)^{\lg_2 c}
  u^{1/c'}=o(u^{1/c})$.

The evaluation time is actually dominated by the calls from level $\ell-1$.
Formally a recursive evaluation from the last recursive 
level $i\leq\ell$ takes time
\begin{align*}
T(\ell)&=O(1)\\
T(i)&=O(c^2/2^i)+d_{(i)} T(i+1)\textnormal{ for }i=0,\ldots,\ell-1
\end{align*}
Our evaluation time is thus $T(0)=O(cD)=o(c^{\lg c})$.

The probability that any of the universal $h_{(i)}$ is not $k$-unique
is bounded by $\sum_{i=0}^{\ell-1}o(|\Phi|^2/\Aout_{(i)}^6)=
o(|\Phi|^2/\Aout_{(\ell-1)}^6)=o(1/|\Phi|^4)=o(1/u^{1/c})$.
\end{proof}
Let $k=u^{\Omega(1/c)}$ be the independence obtained in the above
proof. Consider the  $\ell-1$ recursive levels.
They compose into a function $f:U\fct \Ain^D$, and
we know that $r\circ f$ is $k$-independent. The interesting
point here is that we do not expect $f$ to be $k$-unique.

In \cite[Proposition 2]{thorup12kwise}, or identically, \cite[Theorem 3]{KW12},
is given an exact characterization of the functions $f:U\fct \Ain^D$
that yield $k$-independence when composed with random simple tabulation hashing
$h':\Ain^D\fct R$. The requirement is that every set $X$ of size at most $X$
has some output character appearing an odd number of times. Our $f$ must
satisfy this {\em $k$-odd\/} property. 

We can also make a direct proof that our $f$ is $k$-odd. Noting that
any $k$-unique function is $k$-odd, we just have to
apply the following lemma to our
recursion:
\begin{lemma}
Consider $k$-odd functions $f:U\fct \Ain^c$ and 
$g_i:\Ain\fct\Aout^{d}$, $i\in [c]$.
Let $F:U\fct \Aout^{c\times d}$ be defined
by $F(x)_{(i,j)}=g_i(f(x)_i)_j$. Then $F$ is $k$-odd.
\end{lemma}
\begin{proof}
Consider some set $X$ of size at most $k$. Since $f$ is $k$-odd, 
there is some $i\in[c]$ so that some $(i,a)\in [c]\times \Ain$ 
appears an odd number of times when $f$ is applied to $X$. Let $Y_i$ be the
set of $a\in \Ain$ for which $(i,a)$ is one of
the output position characters that appears an odd number of
times when $f$ is applied to $X$. 
Trivially $|Y_i|\leq k$, so there is an output position 
character $(j,b)\in [d]\times\Aout$ that appears an odd number of times
when $g_i$ is applied to $Y_i$. Then $((i,j),b)$ must also 
appear an odd number of times when $F$ is applied to $X$.
\end{proof}

\section{Counting with care}\label{sec:care}
Over the next two sections, we are now going tighten the analysis from 
Section \ref{sec:basic}. In particular, this will allow us to derive the concrete values from 
Theorem\tref{thm:practice} with no reference to asymptotics. As in Section \ref{sec:basic}, we parameterize our analysis by the length $\ell$ of the list
$L$ of input characters. Later we will add up over relevant lengths $\ell\leq \ell^*=kc$. Using Theorem\tref{thm:key-main}
we fix $q=\eps d/(2c)$ with $\eps=1$ if we are only
interested in uniqueness.

\paragraph{Removing order from the list}
Our first improvement is to argue that we do not need store the order
of the list $L$, i.e., we can just store $L$ as a set. This
immediately saves us a factor $\ell!$, that is, $\choicesl(L)\leq
(c|\Ain|)^\ell/ \ell!<(ec|\Ain|/\ell)^\ell$.  

With $L$ only a
set, an equation $(A,j)$, $A\subseteq L$, $j\in[d]$ still
has the same denotation that $h(A)_j=0$. However, it was the
ordering of $L$ that determined the specified output index
$(\alpha,j)$ with $\alpha$ being the last element
from $A$ in $L$. Changing the ordering of $L$ thus changes the
specified output indices. This may be OK as long as no two equations 
from $M$ specify the same output index.

When $L$ is only given as an unordered set and when we are further given
a set $M$ of equations, we implicitly assign $L$ the lexicographically
first order such that no two equations from $M$ specify the same
output index. This lexicographically first order replaces original order of
$L$ that we found in Section \ref{sec:findingL}. It redefines
the set $I$ of output indices specified by $M$, hence
the set $(L\times[d])\setminus I$ of output indices that have to
be covered by the table $H$.

\paragraph{Equation count}
We will now give a better bound on the number $\choicesl(M)$ of possibilities for
our set $M$ of $\mq$ equations. 
We know that our equations are of the form  $(A,j)$ where $A=x\symmdiff y\subseteq L$ 
for keys $x,y\in X$. More specifically, we have $A=x\symmdiff y=
\{(x_i,i),(y_i,i)|i\in [c], x_i\neq y_i\}$. Let $L_i$ be the set of
input position characters from $L$ in position $i$ and let $\ell_i$ be their
number. Let us assume for now that $\ell_i\geq 2$ for all $i\in[c]$.
If this is not the case, we will later derive even better bounds with
$c'<c$ active positions.

To describe $A$, for each $i\in [c]$, we pick either
two elements from $L_i$ or none. Since $\ell_i\geq 2$, this
can be done in ${\ell_i\choose 2}+1\leq \ell_i^2/2$ ways. The
total number of ways we can pick $A$ is thus
\[\choiceslc(A)\leq \prod_{i\in[c]}\ell_i^2/2\leq 
((\ell/c)^2/2)^{c}.\]
For an equation, we also need $j\in [d]$. We need $\mq\geq \ell \eps d/(2c)$ equations for $M$. 
We conclude that
\begin{align*}
\choiceslc(M)
&\leq {\choiceslc(A)\cdot d \choose \mq}
\leq \left(\frac{e((\ell/c)^2/2)^c d}{\mq}\right)^\mq\\
&\leq \left(\frac{e((\ell/c)^2/2)^c d}{\ell \eps d/(2c)}\right)^\mq
\leq \left(e(\ell/c)^{2c-1}/(\eps 2^{c-1})\right)^\mq.
\end{align*}
Plugging our new bounds into \req{eq:basic}, we get
\begin{align*}
\Plc&\leq 
\frac{\choicesl(L)\cdot \choiceslc(M)}{|\Aout|^{\mq}}\\
&\leq (ec|\Ain|/\ell)^\ell\cdot \left(e(\ell/c)^{2c-1}/(\eps 2^{c-1}|\Aout|)\right)^{\mq}\\
&\leq \left((ec|\Ain|/\ell)\cdot \left(e(\ell/c)^{2c-1}/(\eps 2^{c-1}|\Aout|)\right)^{q}\right)^\ell\\
&=\left((ec|\Ain|/\ell)\cdot \left(e(\ell/c)^{2c-1}/(\eps 2^{c-1}|\Aout|)\right)^{\eps d/(2c)}\right)^\ell.
\end{align*}
As with \req{eq:cheat}, we note that replacing the exponent $\mq$ with
$ q\ell$ is valid whenever the probability bound is not bigger than
$1$.  
Above we assumed that $L$ contained two characters in all $c$ positions. 
In particular, this implies $\ell\geq 2c$.
If this $L$ contains less than two characters in some position, then that position has no effect. There
are ${c\choose c'}$ ways that we can choose $c'$ active positions,
so the real bound we get for $\ell\leq \ell^*$ is:
\[\sum_{c'\leq c}\left({c \choose c'} \sum_{\ell=2c'}^{\ell^*}P_{\ell,c'}\right)\]
The above bound may look messy, but it is easily evaluated by
computer. For $k$-independence, we just need $k$-uniqueness, so
$\eps=1$.  As an example, with 32-bit keys, $c=2$,
$d=20$, and $\ell^*=32$, we get a
total error probability less than $2.58\times 10^{-31}$. This
only allows us to rule out $k=\ell^*/c=16$, but in the next section,
we shall get up to $k=100$ with an even better error probability.

\section{Coding keys}\label{cd:key-codes}
Trivially, we have $\ell\leq kc$ since
$kc$ is the total number of input position characters in the $k$ at most keys. 
However, as $\ell$ approaches
$k$, we can start considering a much more efficient encoding, for
then, instead of encoding equations by the involved input position characters,
we first encode the $k$ keys, and then get a much cheaper encoding of
the equations.

Our goal is to find efficient encodings of symmetric differences
$x\symmdiff y\subseteq L$ where $x,y\in X$. We would like to 
use $L$ to code all the keys in $X$. With that done, to
describe $x\symmdiff y$, we just need to reference $x$ and $y$. A small technical issue is that
$x\in X$ may contain characters not in $L$. As in Section
\ref{sec:care}, we assume that each position $i\in [c]$
is active with least two input position characters $(i,a)\in L$. 
Out of these, we pick a default position character
$(i,a_i)$.  Now if key $x$ contains $(i,x_i)\not\in L$, we replace it
with $(i,a_i)$.  Let $x'$ the result of making these default
replacements of all input position characters outside $L$. Then $x'\subseteq L$.
Moreover, given any two keys $x,y\in X$, if $x\symmdiff y\subseteq L$,
then $x'\symmdiff y'=x\symmdiff y$. 
Each key $x'$ is now described with $c$ characters from $L$, one for each
position, and we get the maximum number of combinations if there
are the same number in each position, so there are at most
$(\ell/c)^c$ combinations. 

Instead of just coding $X'=\{x'\;|\; x\in X\}$, for simplicity, we code
a superset $Y'\supseteq X'$ with exactly $k$ keys. The number of possible $Y'$ is bounded by 
\[{(\ell/c)^c \choose k}<
\left(\frac{e(\ell/c)^c}k\right)^k.\]
An equation is now characterized by two keys from $X'$ and an output
position $j\in [d]$, leaving us ${k\choose 2}d<k^2d/2$ possibilities. We can
therefore pick $\ceil{q\ell}$ equations in less than
\[{k^2d/2\choose \mq}<\left(\frac{e(k^2d/2)}{\mq}\right)^{\mq}
\leq \left(\frac{e(k^2c)}{\eps \ell}\right)^{\mq}\]
ways. Our probability bound for 
a given $\ell$ is thus
\begin{equation}\label{eq:key-codes}
Q^k_{\ell,c}=(ec|\Ain|/\ell)^\ell\left(\frac{e(\ell/c)^c}k\right)^k \left(\frac{e(k^2c)}{\eps \ell|\Aout|}\right)^{\mq}\leq
(ec|\Ain|/\ell)^\ell\left(\frac{e(\ell/c)^c}k\right)^k \left(\frac{e(k^2c)}{\eps \ell|\Aout|}\right)^{\eps d\ell/(2c)}. 
\end{equation}
We are always free to use the best of our probability bounds, so
with $c$ active positions, the probability of getting a list $L$ of size $\ell$
is bounded by $\min\{P_{\ell,c},Q^k_{\ell,c}\}$. Allowing $c'\leq c$ active
positions and considering all lengths $\ell=2c',\ldots,kc'$, we get the overall probability bound
\begin{equation}\label{eq:practice}
\sum_{c'=1}^c\left({c \choose c'} \sum_{\ell=2c'}^{kc'}\min\{P_{\ell,c'},Q^k_{\ell,c'}\}\right).
\end{equation}

\begin{proof}[ of Theorem {\ref{thm:practice}}]
To prove Theorem\tref{thm:practice}, we used $\eps=1$ for uniqueness, and evaluated the sum
\req{eq:practice} with the concrete parameters using a computer.
\end{proof}

We note that we can easily boost the confidence that a random $h$ is
$k$-unique. The basic point is that a lot of the error probability
comes from short codes. With a computer we could, for example,  
check all codes derived from lists $L$ of length $\ell\leq 4$. If
none of these match $h$, then we could start the sum in \req{eq:practice}
from $\ell=5$.

\drop{\paragraph{Acknowledgment} I would like to thank 
Mihai P\v{a}tra\c{s}cu. He suggested the use of a coding based proof instead of my original counting based proof.}


\appendix
\section{From uniqueness to independence}\label{app:siegel}
In this appendix, for completeness, we prove Lemma \ref{lem:siegel}:
\begin{quote}{\em
Let $f:U\fct \Aout^d$ be a $k$-unique function.  Consider a simple
tabulation function $h':\Aout^d\fct R$ where the character tables
$h'_j:\Aout\fct R$, $j\in[d]$, are independent of each other, and
where each $h'_j$ is $k$-independent. Then $h'\circ f:U\rightarrow R$
is $k$-independent.  
}\end{quote} 
The proof is basically the same as
that of Siegel's Lemma 2.6 and Corollary 2.11 in \cite{siegel04hash}, but we need to be
position sensitive. This is a trivial transformation, but since the
proof is so easy, we include it to make our presentation self-contained.
\begin{proof}[ of Lemma \ref{lem:siegel}] 
Take any set $X$ of size at most $k$. We want to show
that $h'\circ f$ hash the keys from $X$ independently. 
First assume that all the character tables $h'_j$ are completely random.
By $k$-uniqueness of $f$, some key $x\in X$
has a unique output position character $(j,f(x)_j)$. Then $x$ is the only key
whose hash $h'(f(x))$ involves the character table value $h'_j(f(x)_j)$. 
If we now fix all other character table values, hence
the hashes of all other keys in $X$, then a uniform
choice of $h'_j(f(x)_j)$ fixes $h'(f(x))$ uniformly.
The hash
of $x$ is therefore independent of the other hashes from $X$. The independence between the remaining hashes values from $X$ follows by induction. The
intermediate
vectors from $f(X)$ have only one character at each position, so the hashing
of $X$ involves at most $k$ values from each $h'_j$. It therefore suffices
that each $h'_j$ is $k$-independent.
\end{proof}

\label{last_page}
\end{document}